\documentclass[a4paper,onecolumn,11pt,accepted=2025-02-10]{quantumarticle}
\pdfoutput=1

\usepackage[utf8]{inputenc}
\usepackage[T1]{fontenc}
\usepackage{hyperref}

\PassOptionsToPackage{usenames, dvipsnames}{color}

\usepackage{lipsum}
\usepackage{array}

\usepackage{enumitem}

\usepackage{ifthen}

\newcommand\curRevision{1}

\newcommand{\revdtext}[2][1]{\ifthenelse{\equal{#1}{\curRevision}}{\textcolor{blue}{#2}}{#2}}

\usepackage{graphicx}
\usepackage{subcaption}
\usepackage{float}

\usepackage{tikz}
\usepackage{pgfplots}
\pgfplotsset{compat=1.17}

\usetikzlibrary{positioning,angles,quotes}
\usetikzlibrary{decorations.pathmorphing,arrows.meta}
\usepackage{bbm}
\usepackage{amsmath}
\usepackage{amssymb}
\usepackage{amsthm}
\usepackage{mathtools}

\DeclarePairedDelimiterX\pSet[1]{\{}{\}}{%
  
  #1
}

\ifdefined\skiptheoremsetup\else
    \newtheorem{innercustomgeneric}{\customgenericname}
    \providecommand{\customgenericname}{}
    \newcommand{\newcustomtheorem}[2]{%
      \newenvironment{#1}[1]
      {%
        \renewcommand\customgenericname{#2}%
        \renewcommand\theinnercustomgeneric{##1}%
        \innercustomgeneric
      }
      {\endinnercustomgeneric}
    }
    \newcommand{\definetheorem}[3][]{
      \newtheorem{#2}[#1]{#3}
      \newcustomtheorem{#2restate}{#3}
    }

    \definetheorem{theorem}{Theorem}
    \definetheorem[theorem]{claim}{Claim}
    \definetheorem[theorem]{corollary}{Corollary}
    \definetheorem[theorem]{conjecture}{Conjecture}
    \definetheorem[theorem]{lemma}{Lemma}
    \definetheorem[theorem]{proposition}{Proposition}
    \definetheorem[theorem]{observation}{Observation}
    \definetheorem[theorem]{remark}{Remark}

    \theoremstyle{definition}
    \definetheorem[theorem]{definition}{Definition}
    \definetheorem{protocol}{Protocol}
    \definetheorem[theorem]{algorithm}{Algorithm}
    \definetheorem[theorem]{problem}{Problem}
    \definetheorem[theorem]{example}{Example}
\fi

\newcommand{\C}{\mathbb{C}}

\newcommand{\N}{\mathbb{N}}

\newcommand{\R}{\mathbb{R}}
\newcommand{\Z}{\mathbb{Z}}

\newcommand{\calH}{\mathcal{H}}

\newcommand{\vv}[1]{\mathbf{#1}}

\usepackage{xparse}

\DeclarePairedDelimiterX\pBrackets[1]{[}{]}{%
    
    #1
}

\newcommand{\Prob}[2][]{
 \def\tmp{#1}%
   \ifx\tmp{}
     \operatorname{Pr}\pBrackets*{#2}
   \else
     \operatorname{Pr}_{#1}\pBrackets*{#2}
   \fi
}

\newcommand{\Expect}[2][]{
 \def\tmp{#1}%
   \ifx\tmp{}
     \operatorname{\mathbb{E}}\pBrackets*{#2}
   \else
     \operatorname{\mathbb{E}}_{#1}\pBrackets*{#2}
   \fi
}

\usepackage{physics}
\usepackage{braket}

\usetikzlibrary{quantikz}

\renewcommand{\ket}[1]{| #1 \rangle}
\renewcommand{\bra}[1]{\langle #1 |}
\renewcommand{\braket}[2]{\langle #1 | #2 \rangle}
\renewcommand{\proj}[1]{\ket{#1}\bra{#1}}
\renewcommand{\ketbra}[2]{\ket{#1}\bra{#2}}

\newcommand{\vspan}[1]{\operatorname{span}\{ #1 \}}

\newcommand{\diag}{\operatorname{diag}}

\newcommand{\id}{\mathbbm{1}}

\newcommand{\sL}[2]{{\mathfrak{s l}(2, \mathbb{C})}}

\newcommand{\T}{\mathbb{T}}

\begin{document}

\title{On multivariate polynomials achievable with quantum signal processing}
\author{Lorenzo Laneve}
\affiliation{Faculty of Informatics — Universit\`a della Svizzera Italiana, 6900 Lugano, Switzerland}
\email{lorenzo.laneve@usi.ch}
\orcid{0000-0003-2319-5456}

\author{Stefan Wolf}
\affiliation{Faculty of Informatics — Universit\`a della Svizzera Italiana, 6900 Lugano, Switzerland}
\maketitle

\begin{abstract}
  \noindent Quantum signal processing (QSP) is a framework which was proven to unify and simplify a large number of known quantum algorithms, as well as discovering new ones. QSP allows one to transform a signal embedded in a given unitary using polynomials. Characterizing which polynomials can be achieved with QSP protocols is an important part of the power of this technique, and while such a characterization is well-understood in the case of univariate signals, it is unclear which multivariate polynomials can be constructed when the signal is a vector, rather than a scalar. This work uses a slightly different formalism than what is found in the literature, and uses it to find simpler necessary conditions for decomposability, as well as a sufficient condition --- the first, to the best of our knowledge, proven for a (generally inhomogeneous) multivariate polynomial in the context of quantum signal processing.
\end{abstract}

%
%

\section{Introduction}
Devising quantum algorithms is a crucial task to understand the potential applications of quantum hardware. For this purpose, a set of composable techniques allows to tackle new problems with relative ease. Notable techniques are amplitude-amplification schemes~\cite{brassardQuantumAmplitudeAmplification2002,groverFastQuantumMechanical1996}, linear combination of unitaries~\cite{childsHamiltonianSimulationUsing2012}, phase-estimation methods~\cite{harrowQuantumAlgorithmLinear2009}, quantum walks~\cite{szegedyQuantumSpeedMarkovChain2004,ambainisQuantumWalkAlgorithm2004,apersUnifiedFrameworkQuantum2021}, and techniques based on query complexity~\cite{cornelissenSpanProgramsQuantum2020,childsQuantumDivideConquer2022,belovsTamingQuantumTime2024}.

Recent research developed a novel framework called \emph{quantum signal processing} (QSP for short)~\cite{lowMethodologyResonantEquiangular2016,lowQuantumSignalProcessing2017} which was shown to unify and simplify a good part of the aforementioned techniques, and beyond~\cite{gilyenQuantumSingularValue2019,martynGrandUnificationQuantum2021}. The idea of this framework is quite simple: we have a \emph{signal} $z$ (typically a complex number of unitary modulus) embedded in a single-qubit \emph{signal operator} (e.g., $z$ is contained in the entries of the matrix). By intertwining calls to this unitary with single-qubit operations (independent of $z$), matrix multiplication gives us a state that is a polynomial in $z$. By a simple tweak (involving phase kickback and Jordan's lemma~\cite{jordanEssaiGeometrieDimensions1875}), one can decompose a high-dimensional unitary $U$ into a direct sum of small subspaces, in which $U$ acts like the signal operator. Within each of these subspaces, we can make $z$ coincide with the eigenvalues of $U$ or the singular values of a matrix \emph{block-encoded} (e.g., the top-left block) in $U$. This construction allows to use QSP to transform the eigenvalues or singular values of matrices given as quantum circuits, using the polynomials implemented with the QSP protocol~\cite{lowOptimalHamiltonianSimulation2017,lowHamiltonianSimulationUniform2017,gilyenQuantumSingularValue2019}.

The power of this technique is that any polynomial satisfying some mild conditions is implementable by a QSP construction --- employing a single control qubit ---, and we can simply specify an algorithm in terms of the polynomials we would like to apply. Notable examples of such conceptual simplification are block-encoded matrix inversion~\cite{lloydHamiltonianSingularValue2021}, phase estimation based on binary search~\cite{martynGrandUnificationQuantum2021}, Hamiltonian simulation~\cite{lowHamiltonianSimulationUniform2017,lowHamiltonianSimulationQubitization2019,martynEfficientFullycoherentQuantum2023}, and state preparation~\cite{mcardleQuantumStatePreparation2022,laneveRobustBlackboxQuantumstate2023}.

Given its success, recent efforts push towards efficient numerical computation of the protocol given a desired polynomial~\cite{haahProductDecompositionPeriodic2019,chaoFindingAnglesQuantum2020, dongEfficientPhasefactorEvaluation2021,dongInfiniteQuantumSignal2024,wangEnergyLandscapeSymmetric2022,mizutaRecursiveQuantumEigenvalue2024}, generally composable paradigms for quantum algorithms~\cite{rossiSemanticEmbeddingQuantum2023,rossiModularQuantumSignal2023}, as well as extensions of the QSP ansatz to different algebras~\cite{haahProductDecompositionPeriodic2019,rossiQuantumSignalProcessing2023,motlaghGeneralizedQuantumSignal2024,laneveQuantumSignalProcessing2024,bastidasComplexificationQuantumSignal2024}.

Among the possible extensions of the model, a version of QSP implementing \emph{multivariate} polynomials was proposed~\cite{rossiMultivariableQuantumSignal2022}. In this case, the signal is a vector $\vv{z}$ where each component is encoded in different signal operators, and the protocol chooses which operator to call at each step. Unfortunately, while nearly any univariate polynomial can be implemented by some QSP protocol, there are some multivariate polynomials --- even with a degree as small as $4$ --- that do not admit such decomposition~\cite{nemethVariantsMultivariateQuantum2023}. A multivariate variant of QSP (M-QSP for short) is quite enticing: applications include computation of functions of commuting matrices~\cite{borns-weilQuantumAlgorithmFunctions2023}, and multivariate Monte Carlo estimation (e.g., following the QSVT-based quantum state preparation scheme of~\cite{mcardleQuantumStatePreparation2022}). 

A lack of characterization of the polynomials decomposable with M-QSP protocols limits the usability of this construction, therefore finding necessary and/or sufficient conditions for applicability is crucial. This work aims at giving further progress on this task.

We use a slightly different formulation of QSP (also used in~\cite{laneveQuantumSignalProcessing2024}) to consider a M-QSP protocol acting on a larger, three-dimensional Hilbert space (which we conjecture to be equivalent to the already-known single-qubit protocol for the class of single-qubit QSP polynomials). This enlarged protocol is a bit easier to work with (since no classical choice is involved in the protocol), and we take advantage of this to derive new, easy-to-check conditions for the (im)possibility to construct a given polynomial. We also prove that some polynomials cannot even be approximated to arbitrary precision by a sequence of constructible polynomials.

The remainder of this work is structured as follows: Section~\ref{sec:univariate-qsp} starts by giving an overview of the variants of univariate QSP found across the literature, to show equivalences and connections with the protocols we are going to analyze. We then proceed with multivariate QSP in Section~\ref{sec:multivariate-qsp}, briefly showing the protocol proposed in~\cite{rossiMultivariableQuantumSignal2022,moriCommentMultivariableQuantum2024}, and extending it to its \emph{analytic} variant. Section~\ref{sec:three-dimensional-qsp} then introduces the aforementioned protocol on three dimensions, proving that it contains the first one, and conjecturing (with motivations) that also the converse holds (Section~\ref{sec:three-dimensional-qsp-equivalence}). We then give a simple necessary and sufficient condition for the existence of a single protocol step that lowers the degree of a given polynomial (needed for induction arguments), extending the conditions stated in~\cite{rossiMultivariableQuantumSignal2022,moriCommentMultivariableQuantum2024} (Section~\ref{sec:three-dimensional-nec-suf-one-step}). In Section~\ref{sec:sufficient-condition-full-decomposability}, we prove that \emph{any} bivariate polynomial $\ket{\gamma(a, b)}$ having non-zero coefficients for $1, a^n, b^n$ can be implemented using our three-dimensional version of M-QSP. To the best of our knowledge, this is the first result that provides a sufficient condition for full M-QSP decomposability. We conclude the work by quantifying the inapproximability of a polynomial (i.e., the best possible error guaranteed by an implementable polynomial to approximate a non-implementable one), also giving a more operational way to conclude the impossibility of a given polynomial (Section~\ref{sec:three-dimensional-inapproximability}).

\subsection{Preliminaries}
We denote with $\T$ the complex unit circle, i.e., the set of complex numbers $z$ with $|z| = 1$. A \emph{Laurent} polynomial is a function $P(z) = \sum_{k = -n}^n p_n z^n$, i.e., a polynomial whose terms can also have negative exponents. We sometimes use the term \emph{analytic} polynomial (term taken from complex analysis), to denote usual polynomials with only non-negative exponents. Symbols written in bold represent vectors: given two vectors $\vv{z} = (z_1, \ldots, z_m), \vv{k} = (k_1, \ldots, k_m)$, we write $\vv{z}^\vv{k}$, as a shorthand for $z_1^{k_1} \cdots z_m^{k_m}$. We also use the bra-ket notation to denote not only quantum states, but also any non-normalized (possibly zero) vector of amplitudes. We use $X, Y, Z, H$ to denote the three Pauli matrices and the single-qubit Hadamard gate, respectively. The notation $[d]$ represents the set $\{ 0, \ldots, d - 1 \}$.
\section{Univariate quantum signal processing}
\label{sec:univariate-qsp}

\noindent In quantum signal processing, the aim is to transform a signal embedded in a unitary with a polynomial. In the end, this task boils down to constructing a so-called \emph{polynomial state}.
\begin{definition}
  A \emph{polynomial state} $\ket{\gamma(\vv{z})}$ is a polynomial vector in $\vv{z} \in \T^m$ that satisfies $\braket{\gamma(\vv{z})}{\gamma(\vv{z})} \equiv 1$.
\end{definition}
\noindent Generally, we have two ways to express a polynomial state of dimension $d$ (for now we will restrict ourselves to the case $d = 2$, i.e., single-qubit polynomials): the first one is simply a vector whose entries are polynomials,
\begin{align}
    \ket{\gamma(\vv{z})} = \sum_{x = 0}^{d-1} P_x(\vv{z}) \ket{x}\ ,
\end{align}
where the normalization condition can be rewritten as $\sum_{x = 0}^{d-1} |P_x(\vv{z})|^2 \equiv 1$. This representation is mainly used in applications of QSP. The second one decomposes the state as a single (vector) polynomial
\begin{align*}
    \ket{\gamma(\vv{z})} = \sum_{\vv{k}}\ket{\gamma_{\vv{k}}} \vv{z}^{\vv{k}}\ .
\end{align*}
Here $\ket{\gamma_{\vv{k}}}$ are the coefficients of the polynomial, which are $d$-dimensional vectors. This decomposition will be useful for us in this work, as we will see that many properties are more easily expressed in terms of these vectors. In this section, we recap the univariate case, which is fully characterized and well-understood in the literature~\cite{martynGrandUnificationQuantum2021,haahProductDecompositionPeriodic2019,motlaghGeneralizedQuantumSignal2024}. For a signal $z \in \T$, we define the Laurent \emph{signal operator}
\begin{align*}
    \Tilde{v}(z) = 
    \begin{bmatrix}
        z^{-1} & 0 \\
        0 & z
    \end{bmatrix}
\end{align*}
where we usually omit the dependency on $z$ for simplicity. A QSP protocol now is a sequence of calls to the signal operator $\Tilde{v}$ intertwined with a sequence of unitary operators which we call \emph{signal processing operators}, or simply \emph{processing operators}. Different variants of QSP are possible (allowing construction of polynomials requiring different criteria) by simply choosing the basis of $\Tilde{v}$ and the domain of the processing operators.
\begin{theorem}[Laurent QSP in the $W_z$ convention~\cite{haahProductDecompositionPeriodic2019,martynGrandUnificationQuantum2021}]
    \label{thm:univariate-qsp-wz-laurent}
    Let $\ket{\gamma(z)}$ be a two-dimensional Laurent polynomial state of degree~$n$, i.e.,
    \begin{align*}
        \ket{\gamma(z)} = P(z) \ket{0} + Q(z) \ket{1}
    \end{align*}
    where $P, Q \in \C[z, z^{-1}]$ are Laurent polynomials of degree~$n$. There exists a vector of $n+1$ phases $\phi_0, \phi_1, \ldots, \phi_n \in [0, 2\pi)$ such that
    \begin{align*}
        e^{i\phi_n X} \ \Tilde{v} \ e^{i\phi_{n-1} X} \Tilde{v} \cdots \Tilde{v} \ e^{i\phi_0 X} \ket{0} = \ket{\gamma(z)}
    \end{align*}
    if and only if:
    \begin{enumerate}[label=(\roman*)]
        \item $P(z)$ has real coefficients and $Q(z)$ has imaginary coefficients;
        \item $P, Q$ have parity $n \bmod 2$.
    \end{enumerate}
\end{theorem}
\noindent This was called the $W_z$ in~\cite{martynGrandUnificationQuantum2021}, since the signal operator (which they called $W(\theta)$) can be seen as a $Z$ rotation. The following variant is called $W_x$ for obvious reasons.
\begin{theorem}[Laurent QSP in the $W_x$ convention~\cite{martynGrandUnificationQuantum2021}]
    \label{thm:univariate-qsp-wx-laurent}
    Let $\ket{\gamma(z)}$ be a two-dimensional Laurent polynomial state of degree~$n$, i.e.,
    \begin{align*}
        \ket{\gamma(z)} = P(z) \ket{0} + Q(z) \ket{1}
    \end{align*}
    where $P, Q \in \C[z, z^{-1}]$ are Laurent polynomials of degree~$n$. There exists a vector of $n+1$ phases $\phi_0, \phi_1, \ldots, \phi_n \in [0, 2\pi)$ such that
    \begin{align*}
        e^{i\phi_n Z} \ H \Tilde{v} H \ e^{i\phi_{n-1} Z} H \Tilde{v} H \cdots H \Tilde{v} H \ e^{i\phi_0 Z} \ket{0} = \ket{\gamma(z)}
    \end{align*}
    if and only if:
    \begin{enumerate}[label=(\roman*)]
        \item $P(z) = P(z^{-1})$ and $Q(z) = -Q(z^{-1})$ for every $z \in \T$;
        \item $P, Q$ have parity $n \bmod 2$.
    \end{enumerate}
\end{theorem}
\noindent A simple conjugation of a $W_x$ protocol with a Hadamard gate gives a $W_z$ protocol (and vice versa). Note that the conjugated signal operator
\begin{align*}
    H \Tilde{v} H =
    \frac{1}{2}
    \begin{bmatrix}
        z^{-1} + z & z^{-1} - z \\
        z^{-1} - z & z^{-1} + z
    \end{bmatrix}
    =
    \begin{bmatrix}
        x & -i\sqrt{1 - x^2} \\
        -i\sqrt{1 - x^2} & x
    \end{bmatrix}
\end{align*}
can be rewritten by replacing $x = (z + z^{-1})/2$, giving us the traditional version of quantum signal processing~\cite{lowMethodologyResonantEquiangular2016,lowQuantumSignalProcessing2017}, from which the quantum singular value transformation can be obtained~\cite{gilyenQuantumSingularValue2019,tangCSGuideQuantum2023}. If, instead of $X$ or $Z$ rotations, we allow the signal-processing operators to be arbitrary unitaries in $SU(2)$, then condition (i) in both the previous theorems will lift (they are actually two subalgebras of a bigger algebra~\cite{chaoFindingAnglesQuantum2020,motlaghGeneralizedQuantumSignal2024}). We remark that using arbitrary $SU(2)$ operators does not increase the complexity of the circuit significantly, as any $SU(2)$ element can be rewritten using Euler's decomposition:
\begin{align*}
    U = e^{i\alpha Z} e^{i\beta X} e^{i\gamma Z} \ \ \  \alpha, \beta, \gamma \in [0, 2\pi).
\end{align*}
Indeed, this decomposition, along with the fact that $Z$ rotations commute with $\Tilde{v}$, yields the protocol with double rotations of~\cite{motlaghGeneralizedQuantumSignal2024} (this also suggests that we do not even need the full $SU(2)$ to remove condition (i)).

We now consider a different, but more intuitive version, recently introduced in~\cite{motlaghGeneralizedQuantumSignal2024}, which considers the following \emph{analytic} signal operator
\begin{align*}
    \Tilde{w}(z) =
    \begin{bmatrix}
        1 & 0 \\
        0 & z
    \end{bmatrix}\ .
\end{align*}
By intertwining this operator with processing operators, we obtain a similar result.
\begin{theorem}[Analytic QSP in the $W_z$ convention~\cite{motlaghGeneralizedQuantumSignal2024}]
    \label{thm:univariate-qsp-wz-analytic}
    Let $\ket{\gamma(z)}$ be a two-dimensional polynomial state of degree~$n$, i.e.,
    \begin{align*}
        \ket{\gamma(z)} = P(z) \ket{0} + Q(z) \ket{1}
    \end{align*}
    where $P, Q \in \C[z, z^{-1}]$ are polynomials of degree~$n$. There exists a vector of $n+1$ phases $\phi_0, \phi_1, \ldots, \phi_n \in [0, 2\pi)$ such that
    \begin{align*}
        e^{i\phi_n X} \ \Tilde{w} \ e^{i\phi_{n-1} X} \Tilde{w} \cdots \Tilde{w} \ e^{i\phi_0 X} \ket{0} = \ket{\gamma(z)}
    \end{align*}
    if and only if:
    \begin{enumerate}[label=(\roman*)]
        \item $P(z)$ has real coefficients and $Q(z)$ has imaginary coefficients for every $z \in \T$.
    \end{enumerate}
\end{theorem}
\noindent Notice how condition (ii) does not appear in analytic QSP. Indeed the constraint of definite parity is only due to the fact that in $\Tilde{v}$ there is a distance of two degrees between $z, z^{-1}$. Similarly as with Laurent QSP, if we extend signal operators to any $SU(2)$ operator, then we would lift condition (i), implying that any analytic polynomial state can be obtained. The following result shows that there is no substantial difference between the Laurent and analytic variants.
\begin{lemma}[Analytic-Laurent correspondence]
    \label{thm:laurent-analytic-correspondence}
    Consider a definite-parity Laurent polynomial state~$\ket{\gamma(z)}$
    \begin{align*}
        \ket{\gamma(z)} = \ket{\gamma_{-n}} z^{-n} + \ket{\gamma_{-n+2}} z^{-n+2} + \cdots + \ket{\gamma_{n-2}} z^{n-2} + \ket{\gamma_n} z^n
    \end{align*}
    where $\ket{\gamma_k}$ are the coefficients of the polynomial. Then $\ket{\gamma(z)}$ is implemented by the Laurent QSP protocol
    \begin{align*}
        A_n \Tilde{v} A_{n-1} \Tilde{v} \cdots \Tilde{v} A_0 \ket{0} = \ket{\gamma(z)}
    \end{align*}
    for some sequence $\{ A_k \}_k \in SU(2)$ of processing operators if and only if the analytic polynomial
    \begin{align*}
        \ket{\gamma^a(z)} = \ket{\gamma_{-n}} + \ket{\gamma_{-n+2}} z + \cdots + \ket{\gamma_{n-2}} z^{n-1} + \ket{\gamma_n} z^n
    \end{align*}
    is implemented by the analytic QSP protocol
    \begin{align*}
        A_n \Tilde{w} A_{n-1} \Tilde{w} \cdots \Tilde{w} A_0 \ket{0} = \ket{\gamma^a(z)}.
    \end{align*}
\end{lemma}
\begin{proof}
    By the fact that $\Tilde{v}(z) = z^{-1} \Tilde{w}(z^2)$, replacing the former with the latter in the protocol gives the Laurent polynomial $z^{-n} \ket{\gamma^a(z^2)}$. If we multiply by $z^n$ and replace $\Tilde{w}(z^2)$ with $\Tilde{w}(z)$, we obtain the desired polynomial. The converse can be obtained by following the argument in reverse.
\end{proof}
\begin{table}
    \centering
    \begin{tabular}{|p{2.5cm} | p{5.5cm} | p{5.5cm}|} 
     \hline
     \rule{0pt}{13pt}
     & \multicolumn{1}{|c|}{\textbf{Laurent picture}~\cite{haahProductDecompositionPeriodic2019}} & \multicolumn{1}{|c|}{\textbf{Analytic picture}~\cite{motlaghGeneralizedQuantumSignal2024}} \\ [0.5ex] 
     & \multicolumn{1}{|c|}{$\Tilde{v} = \diag(z^{-1}, z)$} & \multicolumn{1}{|c|}{$\Tilde{w} = \diag(1, z)$} \\ [1ex]
     \hline\hline
     {\vspace*{0.2cm} \begin{center}$W_x$ convention\end{center}} &
        \begin{enumerate}[label=(\roman*)]
            \item $P(z) = P(z^{-1})$
            \item[] $Q(z) = -Q(z^{-1})$
            \item $P, Q$ have parity $n \bmod 2$
        \end{enumerate} &
        \begin{enumerate}[label=(\roman*)]
            \item $P(z) = z^n P(z^{-1})$
            \item[] $Q(z) = -z^n Q(z^{-1})$
        \end{enumerate} \\
     \hline
     {\vspace*{0.2cm} \begin{center}$W_z$ convention\end{center}} &
        \begin{enumerate}[label=(\roman*)]
            \item $P(z) \in \R[z, z^{-1}]$
            \item[] $Q(z) \in i\R[z, z^{-1}]$
            \item $P, Q$ have parity $n \bmod 2$
        \end{enumerate} &
        \begin{enumerate}[label=(\roman*)]
            \item $P(z) \in \R[z]$
            \item[] $Q(z) \in i\R[z]$
        \end{enumerate} \\ 
     \hline
     {{\vspace*{-0.1cm} \begin{center}Full algebra\end{center}}} &
        \begin{enumerate}[label=(\roman*)]
            \item $P, Q$ have parity $n \bmod 2$
        \end{enumerate} &
        \begin{enumerate}
            \item[] No additional constraints.
        \end{enumerate} \\
     \hline
    \end{tabular}
    \caption{Overview of the variante of univariate QSP found across the literature, with the full characterizations of the polynomial states $\ket{\gamma(z)} = P(z) \ket{0} + Q(z) \ket{1}$.}
    \label{tbl:univariate-qsp-variants}
\end{table}

\noindent Lemma~\ref{thm:laurent-analytic-correspondence} gives a clear one-to-one correspondence between the classes of analytic and Laurent polynomials achievable with QSP, i.e., understanding polynomials of either variant immediately implies a clear characterization on the other. This result is important as it allows us to work with the analytic variant, which is a bit easier to visualize in the subsequent sections: a useful advantage is that the degree of an analytic polynomial state being constructed with a QSP protocol only grows on one side. We remark that conditions (i) in the claims above will be preserved/converted to the counterpart in the other variant (see Table~\ref{tbl:univariate-qsp-variants} for a overview of the univariate versions and their conditions).

\section{Multivariate quantum signal processing}
\label{sec:multivariate-qsp}

\noindent We now talk about the multivariate version of quantum signal processing. Here, the polynomial states $\ket{\gamma(\vv{z})}$ we want to construct transform a signal $\vv{z} \in \T^m$.

\begin{protocol}[Laurent multivariate QSP with classical choice~\cite{rossiMultivariableQuantumSignal2022,moriCommentMultivariableQuantum2024}]
    \label{def:rossi-chuang-qsp}
    Given a vector of phases $\vv{z} \in \T^m$, let $\Tilde{v}_k = \Tilde{v}(z_k) = \diag(z_k^{-1}, z_k)$. Fixing a vector of choices $\vv{s} \in [m]^n$, the following protocol constructs a Laurent multivariate polynomial state
    \begin{align*}
        A_n \Tilde{v}_{s_n} A_{n-1} \Tilde{v}_{s_{n-1}} \cdots \Tilde{v}_{s_1} A_0 \ket{0} = \ket{\gamma(\vv{z})} = P(\vv{z}) \ket{0} + Q(\vv{z}) \ket{1}
    \end{align*}
\end{protocol}
\noindent In this protocol, we follow exactly the same idea as in the univariate case, where at each step we choose a variable~$z_k$ and we apply the signal operator~$\Tilde{v}(z_k)$. Let $n_k$ be the number of times $k$ appears in $\vv{s}$. Necessary conditions for a polynomial state $\ket{\gamma(\vv{z})}$ to be implemented with Protocol~\ref{def:rossi-chuang-qsp} are~\cite{moriCommentMultivariableQuantum2024}:
\begin{enumerate}[label=(\roman*)]
    \item for every $k \in [m]$, the degree of $\ket{\gamma(\vv{z})}$ with respect to $z_k$, i.e., the maximum power of $z_k$ with a non-zero coefficient vector, is at most $n_k$;
    \item for every $k \in [m]$, $\ket{\gamma(\vv{z})}$ has parity $n_k \bmod 2$.
\end{enumerate}
\noindent These conditions need to be satisfied even if the processing operators $A_k$ are arbitrary $SU(2)$ operations. Furthermore, as in the univariate case, we have additional constraints if we restrict to a subset of signal processing operators:
\begin{itemize}
    \item[(iii$^z$)] in the $W_z$ convention, $P$ must have real coefficients, and $Q$ must have imaginary coefficients;
    \item[(iii$^x$)] in the $W_x$ convention, $P(\vv{z}) = P(\vv{z}^{-\vv{1}})$ and $Q(\vv{z}) = -Q(\vv{z}^{-\vv{1}})$ for every choice of $\vv{z} \in \T^m$.
\end{itemize}

\noindent Moreover, it is possible to replace $\Tilde{v}$ in Protocol~\ref{def:rossi-chuang-qsp} with $\Tilde{w}$, thus obtaining an analytic version of this multivariate protocol.
\begin{protocol}[Analytic multivariate QSP with classical choice]
    \label{def:rossi-chuang-qsp-analytic}
    Given a vector of phases $\vv{z} \in \T^m$, let $\Tilde{w}_k = \Tilde{w}(z_k) = \diag(1, z_k)$. With a vector of choices $\vv{s} \in [m]^n$, the following protocol constructs an analytic multivariate polynomial state
    \begin{align*}
        A_n \Tilde{w}_{s_n} A_{n-1} \Tilde{w}_{s_{n-1}} \cdots \Tilde{w}_{s_1} A_0 \ket{0} = \ket{\gamma(\vv{z})} = P(\vv{z}) \ket{0} + Q(\vv{z}) \ket{1}
    \end{align*}
\end{protocol}
\noindent In the spirit of Lemma~\ref{thm:laurent-analytic-correspondence}, whose argument is easily extended to the multivariate case, we will focus on characterizing this version.
\section{A larger protocol}
\label{sec:three-dimensional-qsp}

\noindent Unfortunately, the necessary conditions listed in the previous sections are not sufficient, as counterexamples in~\cite{nemethVariantsMultivariateQuantum2023} show. Understanding which conditions are sufficient for implementability is certainly not an easy task. A~condition highlighted in~\cite{moriCommentMultivariableQuantum2024} for the $W_z$ convention, is the following:
\begin{enumerate}
    \item[(iv$^z$)] Let $\vv{z}_{-k}$ be obtained from $\vv{z}$ by removing $z_k$, and let $P_m(\vv{z}_{-k})$ be the coefficient of $P(\vv{z})$ of $z_k^m$ (analogously for $Q_m(\vv{z}_{-k})$). There exists a $k$ such that
    \begin{align*}
        P_{n_k}(\vv{z}_{-k}) \equiv e^{2\pi i\varphi} Q_{n_k}(\vv{z}_{-k})\ .
    \end{align*}
\end{enumerate}
If such condition holds, then one can choose $\vv{s}_n = k$ and the angle for the $Z$ rotation in the processing operator to be~$\varphi$. Undoing this last step on $\ket{\gamma(\vv{z})}$ will give a valid polynomial state with lower degree $\ket{\gamma'(\vv{z})}$, and the rest of the protocol could be extracted by induction. The problem, however, is that condition (iv) alone on $\ket{\gamma(\vv{z})}$ guarantees conditions (i)-(iii) to hold also for $\ket{\gamma'(\vv{z})}$, but it does not preserve condition (iv).

A huge problem of condition (iv) (besides the fact that it is only valid in the $W_z$ convention) is that it is rather complicated to work with, as it requires the existence of \emph{some} index $k$. Moreover, assuming that we have a polynomial that admits a protocol and we find at some point that more than one choice for $k$ satisfies condition (iv), it is unclear whether any choice taken at this step will certainly lead to a full decomposition. In other words, how can we be sure that there are no ``bad choices''?

Here we circumvent these structural problems by considering a slightly extended protocol (throughout this work we consider only two variables for simplicity, but we can in principle extend it to an arbitrary number).
\begin{protocol}
    \label{def:three-dimensional-qsp}
    Consider the following signal operator
    \begin{align*}
        \Tilde{W} = \diag(1, a, b)
    \end{align*}
    We intertwine this operator with a sequence of processing operators $A_k \in SU(3)$ such that
    \begin{align*}
        A_n \Tilde{W} A_{n-1} \Tilde{W} \cdots \Tilde{W} A_0 \ket{0} = \ket{\gamma(a, b)}.
    \end{align*}
\end{protocol}
\noindent Clearly, such protocol produces polynomial states of dimension $3$ in general:
\begin{align}
    \ket{\gamma(a, b)} = P(a, b) \ket{0} + Q(a, b) \ket{1} + R(a, b) \ket{2} \label{eq:three-dimensional-poly-state}
\end{align}

\subsection{Equivalence with the two-dimensional protocol}
\label{sec:three-dimensional-qsp-equivalence}
\noindent We want to restrict our attention to the polynomials of the form (\ref{eq:three-dimensional-poly-state}) satisfying $R(a, b) \equiv 0$. We give a formal and more comprehensive definition of this class of polynomials, which will be useful later.
\begin{definition}[Effective dimension]
    A polynomial state $\ket{\gamma(\vv{z})} = \sum_{\vv{k}} \ket{\gamma_{\vv{k}}} \vv{z}^{\vv{k}}$ has \emph{effective dimension} $d$ if its coefficient vectors $\ket{\gamma_{\vv{k}}}$ span a subspace of dimension $d$. The same subspace is spanned by $\{ \ket{\gamma(\vv{z})} \}_{\vv{z}}$.
\end{definition}
\noindent States with $R(a, b) \equiv 0$ have effective dimension $\le 2$. In fact, they are the only ones, up to a unitary transformation. Protocol~\ref{def:rossi-chuang-qsp-analytic} is a strict subset of Protocol~\ref{def:three-dimensional-qsp}: we simply ignore the third dimension, choosing only processing operators of the form
\begin{align}
    A_k =
    \begin{bmatrix}
        A_k' & 0 \\
        0 & 1
    \end{bmatrix}
    \label{eq:constrained-processing-operators-3d}
\end{align}
where $A_k' \in SU(2)$ is the processing operator used in the two-dimensional protocol. In this way, $\Tilde{W}$ will act as $\diag(1, a)$ in $\vspan{ \ket{0}, \ket{1} }$. Whenever we choose the signal operator for $b$ at some step, we simply use a permutation operation $S = \proj{0} + \ketbra{1}{2} + \ketbra{2}{1}$ so that $S \Tilde{W} S = \diag(1, b, a)$. Constraining the processing operators in this way certainly gives $R(a, b) \equiv 0$. We conjecture that also the converse is true.
\begin{conjecture}
    \label{thm:conjecture-bidirectional-3d-2d}
    Let $\ket{\gamma(a, b)} = A_n \Tilde{W} A_{n-1} \Tilde{W} \cdots \Tilde{W} A_0 \ket{0}$ be a construction following Protocol~\ref{def:three-dimensional-qsp}.
    If $R(a, b) = \braket{2}{\gamma(a, b)} \equiv 0$, there exist $A_k' \in SU(2)$ and permutation matrices $S_k$ such that
    \begin{align*}
        \begin{bmatrix}
            A_n' & 0 \\
            0 & 1
        \end{bmatrix}
        S_n \Tilde{W} S_n
        \begin{bmatrix}
            A_{n-1}' & 0 \\
            0 & 1
        \end{bmatrix}
        S_{n-1} \Tilde{W} S_{n-1}
        \cdots
        S_1 \Tilde{W} S_1
        \begin{bmatrix}
            A_0' & 0 \\
            0 & 1
        \end{bmatrix}
        \ket{0} = \ket{\gamma(a, b)}\ .
    \end{align*}
\end{conjecture}
\noindent The resulting two-dimensional protocol is thus a sequence of the $SU(2)$ operators $A_k'$ intertwined with calls to $\Tilde{w}_a = \diag(1, a), \Tilde{w}_b = \diag(1, b)$, or $\Tilde{w}_{ab} = \diag(a, b)$, and we simply split $\Tilde{w}_{ab} = X \Tilde{w}_a X \Tilde{w}_b$ to get the final construction compliant with Protocol~\ref{def:rossi-chuang-qsp-analytic}. To see why Conjecture~\ref{thm:conjecture-bidirectional-3d-2d} is not straightforward to prove, consider the following example:
\begin{align*}
    \begin{bmatrix}
        1 \\
        0 \\
        0
    \end{bmatrix}
    \stackrel{A_0}{\mapsto}
    \frac{1}{\sqrt{3}}
    \begin{bmatrix}
        1 \\
        1 \\
        1
    \end{bmatrix}
    \stackrel{\Tilde{W}}{\mapsto}
    \frac{1}{\sqrt{3}}
    \begin{bmatrix}
        1 \\
        a \\
        b
    \end{bmatrix}
    \stackrel{S\Tilde{W}S}{\mapsto}
    \frac{1}{\sqrt{3}}
    \begin{bmatrix}
        1 \\
        ab \\
        ab
    \end{bmatrix}
    \stackrel{A_2}{\mapsto}
    \frac{1}{\sqrt{3}}
    \begin{bmatrix}
        1 \\
        \sqrt{2} ab \\
        0
    \end{bmatrix}\ .
\end{align*}
\noindent The problem with this protocol is that the polynomial does not keep effective dimension $2$ throughout the evolution, and thus cannot be directly translated to Protocol~\ref{def:rossi-chuang-qsp-analytic}. This, however, does not mean this is the only protocol achieving this polynomial; indeed the following does the trick:
\begin{align*}
    \begin{bmatrix}
        1 \\ 0
    \end{bmatrix}
    \stackrel{A_0'}{\mapsto}
    \frac{1}{\sqrt{3}}
    \begin{bmatrix}
        1 \\ \sqrt{2}
    \end{bmatrix}
    \stackrel{\Tilde{w}_a}{\mapsto}
    \frac{1}{\sqrt{3}}
    \begin{bmatrix}
        1 \\ \sqrt{2} a
    \end{bmatrix}
    \stackrel{\Tilde{w}_b}{\mapsto}
    \frac{1}{\sqrt{3}}
    \begin{bmatrix}
        1 \\ \sqrt{2} a b
    \end{bmatrix}\ .
\end{align*}
As a general intuition, if a protocol applies $a$ and $b$ in superposition to two different subspaces $\calH_a$ and $\calH_b$, in order to return to effective dimension $2$, the three polynomials need to become linearly dependent again at some point, which means that $\calH_a$ will eventually receive also $b$ and vice versa (in other words, a segment of the protocol will act as $a^k b^h \cdot \id$ on $\calH_a \oplus \calH_b$). If this is the case, then the same segment can be implemented by multiplying with $a$ and $b$ in two separate steps, always keeping the state in two dimensions. The final claim we could not prove in order to conclude Conjecture~\ref{thm:conjecture-bidirectional-3d-2d} is these situations are the only possible cases that prevent an immediate translation.
\begin{conjecture}
    \label{thm:low-dim-identity-resolution}
    Let $\ket{\gamma(\vv{z})}, \ket{\gamma'(\vv{z})}$ be polynomial states of effective dimension $\le 2$ such that
    \begin{align*}
        \ket{\gamma(\vv{z})} = A_m \Tilde{W} A_{m-1} \cdots A_1 \Tilde{W} \ket{\gamma'(\vv{z})}
    \end{align*}
    but any intermediate state has effective dimension $> 2$. Then the operator $$A_m \Tilde{W} A_{m-1} \cdots A_1 \Tilde{W}$$ acts as $a^{k} b^{h} \cdot U$ in some pair of subspaces $\calH \rightarrow \calH'$ of dimension $d \ge 2$, for some $U \in SU(d)$, $k, h \in \N$.
\end{conjecture}
\noindent Conjecture~\ref{thm:conjecture-bidirectional-3d-2d} would follow directly, since every unitary of the form $a^k b^h \cdot U$ can be implemented trivially, keeping only polynomials with effective dimension $2$ (here, the construction $U_1^\dag \Tilde{w}^k_a \Tilde{w}^h_b X \Tilde{w}^k_a \Tilde{w}^h_b X U_2$ is an example of such protocol, where $U = U_1^\dag U_2$ such that $U_1, U_2$ map our state within $\vspan{\ket{0}, \ket{1}}$).

We highlight that Conjecture~\ref{thm:low-dim-identity-resolution} becomes trivial when we replace the bound on the effective dimension above, from $2$ to $1$: such state would be of the form $a^k b^h \ket{\psi}$ for some quantum state $\ket{\psi}$ independent of $a, b$. We would necessarily have $\ket{\gamma(\vv{z})} = a^k b^h \ket{\phi}, \ket{\gamma'(\vv{z})} = a^{k'} b^{h'} \ket{\psi}$, with $k' \le k, h' \le h$ ($\Tilde{W}$ cannot decrease the degree of the polynomial state). The claim follows by taking any $U : \ket{\psi} \mapsto \ket{\phi}$.

\subsection{Conditions for extractability of one step}
\label{sec:three-dimensional-nec-suf-one-step}
\noindent We now give some conditions for the existence of an instance of Protocol~\ref{def:three-dimensional-qsp} to construct a given polynomial state $\ket{\gamma(\vv{z})}$. These conditions work for general three-dimensional polynomials, but we remember that in the case such polynomial has effective dimension $2$ we can also retrieve a construction in the sense of Protocol~\ref{def:rossi-chuang-qsp-analytic} (provided Conjecture~\ref{thm:conjecture-bidirectional-3d-2d} holds).

Before we dive into stating necessary and/or sufficient conditions, we first provide an intuition on what happens throughout an evolution of Protocol~\ref{def:three-dimensional-qsp}. We start from the state $A_0 \ket{0}$, which is an arbitrary three-dimensional quantum state (in other words, the coefficient vector $\ket{\gamma_{0, 0}} = A_0 \ket{0}$, and $\ket{\gamma_{\vv{k}}} = 0$ for $\vv{k} \neq (0, 0)$). An application of $\Tilde{W}$ splits the space into the three components $\{ \ket{0}, \ket{1}, \ket{2} \}$ of the computational basis, more explicitly
\begin{align*}
    \Tilde{W}
    \begin{bmatrix}
        \psi_0 \\
        \psi_1 \\
        \psi_2
    \end{bmatrix}
    = 
    \psi_0 \ket{0} + \psi_1 \ket{1} a + \psi_2 \ket{2} b\ ,
\end{align*}
which means that, after one step, the coefficient vectors $\ket{\gamma_{0, 0}}, \ket{\gamma_{1, 0}}, \ket{\gamma_{0, 1}}$ are pairwise orthogonal (the choice of the subsequent processing operator $A_1$ does not alter this fact). We can see the evolution of the coefficient vectors as propagation in a lattice (see Figure~\ref{fig:three-dim-qsp-evolution}): after $n$ steps, the points corresponding to non-zero coefficients form a right triangle with the two short sides of length $n$. Before applying $A_n$, we can see that the coefficient vectors satisfy the following conditions:
\begin{itemize}
    \item $\ket{\gamma_{k, 0}} \in \vspan{ \ket{0}, \ket{1} }$ for $0 \le k \le n$;
    \item $\ket{\gamma_{0, k}} \in \vspan{ \ket{0}, \ket{2} }$ for $0 \le k \le n$;
    \item $\ket{\gamma_{k, n - k}} \in \vspan{ \ket{1}, \ket{2} }$ for $0 \le k \le n$.
\end{itemize}
These three sets of vectors correspond to the three sides of the triangle in Figure~\ref{fig:three-dim-qsp-evolution}. In particular, notice that the three endpoints of the triangle must satisfy two of these conditions, and in particular, they will be pairwise orthogonal.

\begin{figure}
    \centering
    \begin{tikzpicture}[scale=0.9]
    \filldraw[black] (2.5, 0) circle (2pt);
    \node[inner sep=7pt,anchor=south] at (2.5, 0) {$\ket{\gamma_{0, 0}} = \ket{0}$};

    \node at (4.5, -0.5) {\Large $\Tilde{W} A_0$};
    \draw[->, thick] (4, -1) -- (5, -1);

    \filldraw[black] (6.5, 0) circle (2pt);
    \node[inner sep=7pt,anchor=south] at (6, 0) {$\ket{\gamma_{0, 0}} \propto \ket{0}$};
    \filldraw[black] (8.5, 0) circle (2pt);
    \node[inner sep=7pt,anchor=south] at (8.5, 0) {$\ket{\gamma_{1, 0}} \propto \ket{1}$};
    \filldraw[black] (6.5, -2) circle (2pt);
    \node[inner sep=7pt,anchor=west] at (6.5, -2) {$\ket{\gamma_{0, 1}} \propto \ket{2}$};

    \draw[->, thick, red] (7, 0) -- (8, 0) node[pos=0.5, anchor=north] {$\times a$};
    \draw[->, thick, blue] (6.5, -0.5) -- (6.5, -1.5) node[pos=0.5, anchor=west] {$\times b$};

    \node at (10, -0.5) {\Large $\Tilde{W} A_1$};
    \draw[->, thick] (9.5, -1) -- (10.5, -1);

    \filldraw[black] (12, 0) circle (2pt);
    \node[inner sep=7pt,anchor=south] at (12, 0) {$\ket{\gamma_{0, 0}} \propto \ket{0}$};
    \filldraw[black] (14, 0) circle (2pt);
    \node[inner sep=7pt,anchor=south] at (14, 0) {};
    \filldraw[black] (16, 0) circle (2pt);
    \node[inner sep=7pt,anchor=south] at (16, 0) {$\ket{\gamma_{2, 0}} \propto \ket{1}$};
    \filldraw[black] (12, -2) circle (2pt);
    \node[inner sep=7pt,anchor=west] at (12, -2) {};
    \filldraw[black] (14, -2) circle (2pt);
    \node[inner sep=7pt,anchor=west] at (14, -2) {$\ket{\gamma_{1, 1}} \perp \ket{0}$};
    \filldraw[black] (12, -4) circle (2pt);
    \node[inner sep=7pt,anchor=west] at (12, -4) {$\ket{\gamma_{0, 2}} \propto \ket{2}$};

    \draw[->, thick, red] (12.5, 0) -- (13.5, 0) node[pos=0.5, anchor=north] {$\times a$};
    \draw[->, thick, blue] (12, -0.5) -- (12, -1.5) node[pos=0.5, anchor=west] {$\times b$};
    \draw[->, thick, red] (14.5, 0) -- (15.5, 0) node[pos=0.5, anchor=north] {$\times a$};
    \draw[->, thick, blue] (14, -0.5) -- (14, -1.5) node[pos=0.5, anchor=west] {$\times b$};
    \draw[->, thick, red] (12.5, -2) -- (13.5, -2) node[pos=0.5, anchor=north] {$\times a$};
    \draw[->, thick, blue] (12, -2.5) -- (12, -3.5) node[pos=0.5, anchor=west] {$\times b$};
\end{tikzpicture}
    \caption{Evolution of the coefficient vectors of a three-dimensional polynomial state throughout two steps of Protocol~\ref{def:three-dimensional-qsp}. The evolution can be visualized as a propagation on a two-dimensional lattice, where each point $(x, y)$ corresponds to the coefficient vector of $a^x b^y$. Right after applying $\Tilde{W}$, the coefficient vectors on the boundaries are proportional to an element of the computational basis. The processing operators $A_k$ can change the basis of these vectors, but not the inner product between them.}
    \label{fig:three-dim-qsp-evolution}
\end{figure}
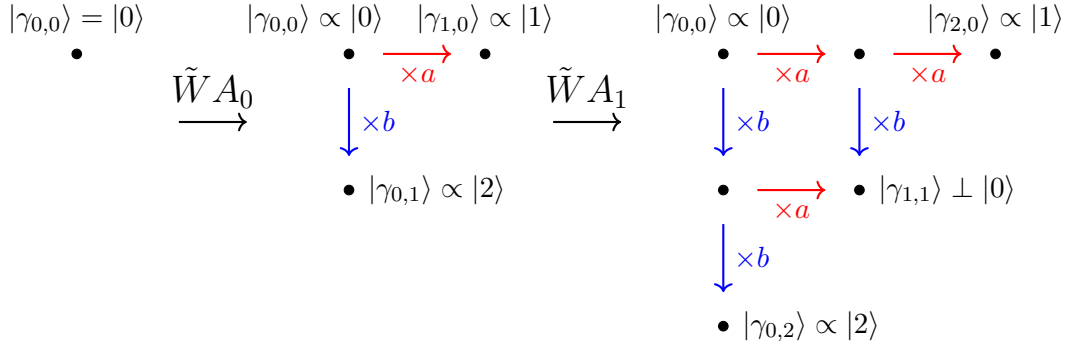

\begin{theorem}[Necessary and sufficient condition for the extraction of one step]
    \label{thm:condition-for-single-step}
    Let $\ket{\gamma(a, b)} = \sum_{k, h} \ket{\gamma_{k, h}} a^{k} b^{h}$ be a three-dimensional polynomial state of degree~$n$. There exists $A_n \in SU(3)$ such that $\Tilde{W}^\dag A_n^\dag \ket{\gamma(a, b)}$ has degree $n - 1$ if and only if there exists an orthonormal basis $\ket{\psi_0}, \ket{\psi_1}, \ket{\psi_2}$ such that
    \begin{enumerate}[label=(\roman*)]
        \item $\braket{\psi_2}{\gamma_{k, 0}} = 0$ for $0 \le k \le n$;
        \item $\braket{\psi_1}{\gamma_{0, k}} = 0$ for $0 \le k \le n$;
        \item $\braket{\psi_0}{\gamma_{k, n - k}} = 0$ for $0 \le k \le n$.
    \end{enumerate}
    If the three endpoints $\ket{\gamma_{0, n}}, \ket{\gamma_{n, 0}}, \ket{\gamma_{0,0}}$ are non-zero, then the conditions above are equivalent to
    \begin{enumerate}[label=(\roman*')]
        \item $\braket{\gamma_{0, n}}{\gamma_{k, 0}} = 0$ for $0 \le k \le n$;
        \item $\braket{\gamma_{n, 0}}{\gamma_{0, k}} = 0$ for $0 \le k \le n$;
        \item $\braket{\gamma_{0, 0}}{\gamma_{k, n - k}} = 0$ for $0 \le k \le n$.
    \end{enumerate}
\end{theorem}
\begin{proof}
    We have already shown condition (i)-(iii) to be necessary, we prove their sufficiency. The polynomial $\ket{\gamma'(a, b)} = \Tilde{W}^\dag A_n^\dag \ket{\gamma(a, b)}$ has degree $n - 1$ if and only if $\ket{\gamma'_{k, n-k}} = 0$ for every $0 \le k \le n$ (higher-degree coefficients will be zero by construction). Moreover, we need $\ket{\gamma'_{-1, k}} = \ket{\gamma'_{k, -1}} = 0$ for $0 \le k \le n$ in order to have a valid analytic polynomial. Since $A_n^\dag \sum_{k, h} \ket{\gamma_{k,h}} a^k b^h = \sum_{k, h} A_n^\dag \ket{\gamma_{k,h}} a^k b^h$, this degree reduction will be guaranteed if the following conditions are true:
    \begin{itemize}
        \item $A_n^\dag \ket{\gamma_{k, 0}} \in \vspan{\ket{0}, \ket{1}}$ for $0 \le k \le n$;
        \item $A_n^\dag \ket{\gamma_{0, k}} \in \vspan{\ket{0}, \ket{2}}$ for $0 \le k \le n$;
        \item $A_n^\dag \ket{\gamma_{k, n-k}} \in \vspan{\ket{1}, \ket{2}}$ for $0 \le k \le n$.
    \end{itemize}
    Taking $A_n: \ket{j} \mapsto \ket{\psi_j}$ satisfies this condition. Furthermore, conditions (i)-(iii) imply the following:
    \begin{align*}
        \ket{\gamma_{0,n}} \propto \ket{\psi_2}, \ \ \ket{\gamma_{n,0}} \propto \ket{\psi_1}, \ \ \ket{\gamma_{0,0}} \propto \ket{\psi_0}
    \end{align*}
    because these three vectors must satisfy two of the three conditions (the only subtlety here is that they might be zero). If all of them are non-zero, then by conditions (i')-(iii') they are pairwise orthogonal, proving that also (i')-(iii') implies (i)-(iii).
\end{proof}
\noindent These extend and simplify the necessary condition given in~\cite{nemethVariantsMultivariateQuantum2023}. In the univariate case, conditions (i)-(iii) boil down to $\braket{\gamma_0}{\gamma_n} = 0$, i.e., the coefficient vectors at the two endpoints must be orthogonal (which is in turn guaranteed by the fact that $\braket{\gamma(z)}{\gamma(z)} \equiv 1$). This allows to apply induction easily, proving Theorems~\ref{thm:univariate-qsp-wz-laurent}-\ref{thm:univariate-qsp-wz-analytic}. Unfortunately, unlike in the univariate case, the conditions of Theorem~\ref{thm:condition-for-single-step} are not sufficient for the existence of a full protocol implementing~$\ket{\gamma(a, b)}$: although $\Tilde{W}^\dag A_n^\dag \ket{\gamma(a, b)}$ has lower degree, it is not guaranteed that the new polynomial in turn satisfies these conditions.

We apply our result to a concrete example (converted to analytic from~\cite{nemethVariantsMultivariateQuantum2023}):
\begin{align*}
    P(a, b) & = a^2 b^2 + 1 - \frac{122 + 8i}{37} (a b^2 + a) + \frac{362 - 248i}{111} (a^2 b + b) + \\
    & + \frac{114 + 56i}{37} (a^2 + b^2) + \left(\frac{692}{111} - \frac{719i}{222}\right) ab \\
    Q(a, b) & = a^2 b^2 - 1 - \frac{122 + 66i}{37} (a b^2 - a) + \\ 
    & - \frac{56 + 114i}{37} (a^2 - b^2) + \frac{362 - 418i}{111} (a^2 b - b)
\end{align*}
Although $\ket{\gamma(a,b)} = P(a, b) \ket{0} + Q(a, b) \ket{1}$ is a polynomial state (up to a normalization factor), it does not satisfy the conditions of Theorem~\ref{thm:condition-for-single-step}: the two sets of vectors $\{ \ket{\gamma_{0,1}}, \ket{\gamma_{0,2}} \}, \{ \ket{\gamma_{1,0}}, \ket{\gamma_{2,0}} \}$ are both linearly independent, and they span $\{ \ket{0}, \ket{1} \}$. This means that there cannot exist $\ket{\psi_2}, \ket{\psi_1}$ satisfying conditions (i) and (ii) of Theorem~\ref{thm:condition-for-single-step} (they would need to be both $\propto \ket{2}$). We state the sufficient condition we just used for the non-implementability of a two-dimensional polynomial state more explicitly.
\begin{corollary}
    \label{thm:sufficient-condition-unimplementability}
    Let $\ket{\gamma(a, b)} = P(a, b) \ket{0} + Q(a, b) \ket{1} = \sum_{k, h} \ket{\gamma_{k, h}} a^k b^h$ be a polynomial state of degree~$n$. If if holds that $$\vspan{\ket{\gamma_{k, 0}}}_k = \vspan{\ket{\gamma_{0, k}}}_k = \vspan{\ket{0}, \ket{1}},$$ then $\ket{\gamma(a, b)}$ cannot be implemented with Protocol~\ref{def:three-dimensional-qsp} (and thus, not even with Protocol~\ref{def:rossi-chuang-qsp-analytic}).
\end{corollary}
\noindent This gives a very simple condition to check in order to prove the impossibility. An important remark is that this claim does not need Conjecture~\ref{thm:conjecture-bidirectional-3d-2d} to hold, since we proved that Protocol~\ref{def:rossi-chuang-qsp-analytic} is at least a subset of Protocol~\ref{def:three-dimensional-qsp}. 
\subsection{A sufficient condition for decomposability}
\label{sec:sufficient-condition-full-decomposability}
\noindent We use what we found in the previous section to first show a constructive result, i.e., a \emph{sufficient} condition for the existence of an instance of Protocol~\ref{def:three-dimensional-qsp}.
\begin{theorem}
    \label{thm:sufficient-condition-full-decomposition}
    Let $\ket{\gamma(a, b)} = \sum_{k, h} \ket{\gamma_{k, h}} a^k b^h$ be a polynomial state of degree~$n$ such that the coefficient vectors $\ket{\gamma_{0,0}}, \ket{\gamma_{n,0}}, \ket{\gamma_{0,n}}$, i.e., the coefficients of $1, a^n, b^n$ are non-zero. Then there exist $A_0, \ldots, A_n \in SU(3)$ such that
    \begin{align*}
        A_n \Tilde{W} A_{n-1} \Tilde{W} \cdots \Tilde{W} A_0 \ket{0} = \ket{\gamma(a, b)}\ .
    \end{align*}
\end{theorem}
\begin{proof}
    The claim is trivial for $n = 0$, since any zero-degree polynomial is simply a quantum state $\ket{\psi}$, and any $A_0$ containing $\ket{\psi}$ as first column does the trick. Therefore, let us consider $n > 0$.

    Since $\ket{\gamma_{0,0}}, \ket{\gamma_{n,0}}, \ket{\gamma_{0,n}} \neq 0$, the conditions (i')-(iii') of Theorem~\ref{thm:condition-for-single-step} must hold in order for a unitary $A_n$ to lower the degree. This is actually guaranteed by the normalization condition:
    \begin{align*}
        \braket{\gamma(\vv{z})}{\gamma(\vv{z})} = \sum_{\vv{k}, \vv{h}} \braket{\gamma_{\vv{k}}}{\gamma_{\vv{h}}} \vv{z}^{\vv{h} - \vv{k}} = \sum_{\vv{j}} \sum_{\vv{k}} \braket{\gamma_{\vv{k}}}{\gamma_{\vv{k} + \vv{j}}} \vv{z}^{\vv{j}} \stackrel{!}{\equiv} 1
    \end{align*}
    The coefficient of $\vv{z}^{\vv{j}}$ for $\vv{j} \neq 0$ must be zero. In particular
    \begin{itemize}
        \item by taking $\vv{j} = (k, -n)$, we obtain $\braket{\gamma_{0,n}}{\gamma_{k,0}} = 0$, condition (i');
        \item by taking $\vv{j} = (-n, k)$, we obtain $\braket{\gamma_{n,0}}{\gamma_{0,k}} = 0$, condition (ii');
        \item by taking $\vv{j} = (k, n - k)$, we obtain $\braket{\gamma_{0,0}}{\gamma_{k,n-k}} = 0$, condition (iii').
    \end{itemize}
    Thus, by Theorem~\ref{thm:condition-for-single-step}, there exists a unitary $A_n$ such that $\Tilde{W}^\dag A_n^\dag \ket{\gamma(a, b)}$ has degree $n - 1$. In order to conclude the induction step we only need to prove that $1, a^{n-1}, b^{n-1}$ have a non-zero coefficient in the new polynomial. Consider $\ket{\gamma'(a, b)} = A_n^\dag \ket{\gamma(a, b)}$. By linearity, $A_n^\dag$ is simply applied to each $\ket{\gamma_{k, h}}$, thus the three endpoints are still non-zero after the application of~$A_n^\dag$. By considering $\ket{\gamma'(a, b)} = P(a,b) \ket{0} + Q(a,b) \ket{1} + R(a,b) \ket{2}$, since the degree of $\ket{\gamma'(a, b)}$ will be lowered by $\Tilde{W}^\dag$, this means that
    \begin{itemize}
        \item only $P$ has the constant term, $\ket{\gamma_{0,0}} \propto \ket{0}$;
        \item only $Q$ has the $a^n$ term, $\ket{\gamma_{n,0}} \propto \ket{1}$;
        \item only $R$ has the $b^n$ term, $\ket{\gamma_{0,n}} \propto \ket{2}$.
    \end{itemize}
    As $\Tilde{W}^\dag$ simply shifts these coefficients, we can conclude that the polynomial $\Tilde{W} \ket{\gamma'(a, b)} = \Tilde{W}^\dag A_n^\dag \ket{\gamma(a, b)}$ has degree $n - 1$, with the coefficients of $1, a^{n-1}, b^{n-1}$ being non-zero. This concludes the proof by applying the induction step.
\end{proof}
\noindent Here we treat the case of two variables for simplicity, but it is possible to define a protocol on a $(d+1)$-level system and $d$ variables (we leave the full argument in Appendix~\ref{apx:proof-m-variables} for completeness).

\begin{example}
    Let $\omega = e^{2\pi i/3}$ be the cube root of unity and given $\vec{r} = (r_1, r_2, \ldots, r_n) \in \Z_3^n$, we define its Fourier path as
    \begin{align*}
        f(\vec{r}) = \omega^{r_n r_{n-1} + r_{n-1} r_{n-2} + \cdots + r_2 r_1}
    \end{align*}
    Let $S_{k, h}$ be the set of all paths that have $k$ occurrences of $1$ and $h$ occurrences of $2$. We can define the triple of polynomials:
    \begin{align*}
        P(a, b) & = \sum_{k + h \le n} \qty(\sum_{\vec{r} \in S_{k, h}} f(\vec{r})) a^k b^h \\
        Q(a, b) & = \sum_{k + h \le n} \qty(\sum_{\vec{r} \in S_{k, h}} f(\vec{r}) \omega^{r_n}) a^k b^h \\
        R(a, b) & = \sum_{k + h \le n} \qty(\sum_{\vec{r} \in S_{k, h}} f(\vec{r}) \omega^{2 r_n}) a^k b^h
    \end{align*}
    In other words, the coefficients of $a^k b^h$ are summing the Fourier paths over all the possible choices containing $k$ 1's and $h$ 2's (the phase difference between $P, Q, R$ can represent the $(n+1)$-th choice).
    An induction argument on the paths can show that these polynomials sum up to $1$ (up to an omitted multiplicative constant). Moreover, the sets $S_{n,0}, S_{0,n}$ and $S_{0,0}$ contain exactly one path, and therefore this triple satisfies the conditions of Theorem~\ref{thm:sufficient-condition-full-decomposition}, which then gives an instance of Protocol~\ref{def:three-dimensional-qsp}. It turns out that the processing operator are equal to the quantum Fourier transform over $\Z_3$, while $A_0$ prepares the state $(1, \omega^2, \omega^2)/\sqrt{3}$.
\end{example}
\subsection{Inapproximability}
\label{sec:three-dimensional-inapproximability}
\noindent Another result we can derive is a proof that some polynomial states are not only impossible to be implemented exactly, but they are even impossible to approximate with good precision by a QSP protocol.

\begin{theorem}
    \label{thm:inapproximability-bounds}
    Let $\ket{\gamma(a, b)}$ be a two-dimensional polynomial state of degree~$n$. Consider:
    \begin{align*}
        q(\gamma) = \min\left\{ \max_{0 \le x, y \le n}\bigg| \det\bigg[\ket{\gamma_{x, 0}} \ \ket{\gamma_{y, 0}}\bigg] \bigg|, \max_{0 \le x, y \le n} \bigg| \det\bigg[\ket{\gamma_{0, x}} \ \ket{\gamma_{0, y}}\bigg] \bigg| \right\}\ .
    \end{align*}
    Any polynomial state $\ket{\gamma'(a, b)}$ satisfying $\lVert \ket{\gamma(a, b)} - \ket{\gamma'(a, b)} \rVert < q(\gamma)/2$ for any $a, b \in \T$ cannot be implemented.
\end{theorem}
\noindent This also gives a more practical method to check conditions of Theorem~\ref{thm:sufficient-condition-unimplementability}. Indeed, if $q(\gamma) > 0$ then both spans have dimension $2$, and the polynomial cannot be implemented. However, this does not exclude that some non-implementable polynomials have $q(\gamma) = 0$ (indeed, taken any such polynomial $\ket{\gamma(a, b)}$, then $\ket{\gamma'(a, b)} = A \Tilde{W} \ket{\gamma(a, b)}$, albeit impossible to construct, satisfies $q(\gamma') = 0$ regardless of $q(\gamma)$).
\begin{proof}
    We show that such $\ket{\gamma'(a, b)}$ has $q(\gamma') > 0$. Assuming that $\sup_{a, b \in \T} \lVert \ket{\gamma(a, b)} - \ket{\gamma'(a, b)} \rVert^2 < \epsilon^2$, we also have:
    \begin{align*}
        \frac{1}{(2\pi)^2} \iint_0^{2\pi} \big\lVert \ket{\gamma(e^{i\theta}, e^{i\eta})} - \ket{\gamma'(e^{i\theta}, e^{i\eta})} \big\rVert^2 d\theta \ d\eta < \epsilon^2
    \end{align*}
    Using Parseval's identity~\cite{steinFourierAnalysisIntroduction2011}, we conclude that
    \begin{align*}
        \sum_{k, h} \big\lVert \ket{\gamma'_{k,h}} - \ket{\gamma_{k,h}} \big\rVert^2 < \epsilon^2
    \end{align*}
    and, in particular, this inequality holds for each term of the sum. By the fact that $\abs{\det[a, b]} \le \norm{a} \norm{b}$ (direct consequence of the Cauchy-Schwarz inequality), and that $\det[a, b] - \det[a, c] = \det[a, b-c]$, along with a standard application of the triangle inequality, we obtain, for $0 \le x, y \le n$:
    \begin{align*}
        \bigg| \det\bigg[\ket{\gamma_{x, 0}} \ \ket{\gamma_{y, 0}}\bigg] - \det\bigg[\ket{\gamma'_{x, 0}} \ \ket{\gamma'_{y, 0}}\bigg] \bigg| < 2\epsilon \\
        \bigg| \det\bigg[\ket{\gamma_{0, x}} \ \ket{\gamma_{0, y}}\bigg] - \det\bigg[\ket{\gamma'_{0, x}} \ \ket{\gamma'_{0, y}}\bigg] \bigg| < 2\epsilon
    \end{align*}
    This implies that $|q(\gamma) - q(\gamma')| < 2\epsilon$ and, in particular, $q(\gamma') > 0$ if we take $\epsilon = q(\gamma)/2$.
\end{proof}
\noindent By a simple calculation, the counterexample shown in Section~\ref{sec:three-dimensional-nec-suf-one-step}, has $q(\gamma) = 144/10625 \simeq 0.013$.
\section{Discussion and conclusions}

In this work we consider the problem of multivariate quantum signal processing. Instead of previous work~\cite{rossiMultivariableQuantumSignal2022}, where the signal ($a$ or $b$) to ``blend'' at each step is chosen classically, we define here a protocol acting on a three-dimensional space, where essentially the choice of the variable can be made in superposition, which makes the evolution easier to understand and to decompose.

While Protocol~\ref{def:three-dimensional-qsp} produces three-dimensional polynomials in general, Conjecture~\ref{thm:conjecture-bidirectional-3d-2d} states that, whenever our desired polynomial has effective dimension $2$, then a construction using three dimensions can be turned into a bi-dimensional one (perhaps by increasing its length by at most a factor of $2$). While we show a possible direction of the proof, its gist requires a deeper understanding of what we can do with Protocol~\ref{def:three-dimensional-qsp}. We also remark that, whereas Conjecture~\ref{thm:conjecture-bidirectional-3d-2d} would turn the protocol into a single-qubit one, even if the conjecture turns out to be false, Protocol~\ref{def:three-dimensional-qsp} can still be implemented with two qubits, although we would need to implement general $SU(3)$ unitaries (it is possible to obtain a phase factor decomposition using $8$ rotations by, e.g., the Euler decomposition by the Gell-Mann matrices~\cite{roelfsGeometricInvariantDecomposition2022}, or the Sinkhorn normal form~\cite{idelSinkhornNormalForm2015}). Moreover, we remind that Laurent polynomials are always obtainable from analytic polynomials even in this case, by shifting all the coefficients with an operator $a^{-k} b^{-k} \id$ (in applications such as quantum eigenvalue transformation, this is possible by calling the unitaries for $a, b$ unconditionally --- this trick was also exploited in~\cite{laneveQuantumSignalProcessing2024,berryDoublingEfficiencyHamiltonian2024a}).

We then state necessary or sufficient conditions for decomposability into Protocol~\ref{def:three-dimensional-qsp}: we first give a necessary and sufficient condition (Theorem~\ref{thm:condition-for-single-step}) for the existence of a processing operator $A \in SU(3)$ that lowers the degree of the polynomial (usually needed in induction arguments which aim at finding a decomposition). These conditions come from the intuition that $m$-variable M-QSP protocols can be seen as a propagation in a $m$-dimensional lattice, where each point represents a coefficient of the multivariate polynomial.

These conditions are used to prove a sufficient condition for decomposability, namely that any $n$-degree polynomial having a non-zero constant term, as well as non-zero $a^n, b^n$ terms, can be constructed using Protocol~\ref{def:three-dimensional-qsp}. This is the first constructive result ever proven for a multivariate polynomial in the context of quantum signal processing. This proves that, in some sense, the behaviour of multivariate polynomials closely follows the univariate case, with the only exception that ``degenerate'' univariate schemes (the ones that do not reach the $n$-th degree term with $n$ steps) simply produce polynomials with degree $< n$, while multivariate polynomials without the $a^n$ term are not necessarily of degree $< n$. With similar arguments, an analogous result for polynomials in $d$ variables and a $d+1$ dimensional protocol is given in Appendix~\ref{apx:proof-m-variables}, which also gives a sufficient condition for \emph{homogeneous quantum signal processing} with general number of variables, a problem left open in~\cite{nemethVariantsMultivariateQuantum2023} (indeed, in order to obtain the homogeneous version we simply replace the $1$ in the signal operator with a new variable). We remark that this condition is certainly not necessary, since the polynomial $(1, ab, 0)/\sqrt{2}$ can be easily implemented in two steps (indeed, any polynomial state with effective dimension $2$ cannot satisfy this condition).

We conclude the work by giving a simpler, easy-to-check, necessary condition for decomposability of a bivariate polynomial, involving a quantity $q(\gamma)$. We use such a simpler condition to prove that some polynomials are even hard to approximate with arbitrary precision. This confirms that the set of non-constructible polynomials has non-zero measure, formally proving a claim that was shown numerically in~\cite{nemethVariantsMultivariateQuantum2023}.

A natural question would be whether a multivariate QSVT could arise directly from (a restriction of) Protocol~\ref{def:three-dimensional-qsp}, without having to resort to Conjecture~\ref{thm:conjecture-bidirectional-3d-2d}, perhaps with a three-dimensional extension of the cosine-sine decomposition~\cite{tangCSGuideQuantum2023}. Another question left open is to understand whether there is an extension of the Fej\'er-Riesz theorem~\cite{geronimoPositiveExtensionsFejerRiesz2004,hussenFejerRieszTheoremIts2021} (used also for polynomial completion in univariate QSP~\cite{haahProductDecompositionPeriodic2019} and single-qubit multivariate QSP~\cite{rossiMultivariableQuantumSignal2022}) allowing to complete a given polynomial $P(a,b)$ with a pair of polynomials $Q(a,b), R(a,b)$ satisfying the sufficient condition of Theorem~\ref{thm:sufficient-condition-full-decomposition}. Moreover, it is a possible future direction to understand whether the above conditions can be extended to the case of non-commutative signals, a case relevant for quantum eigenvalue transformation involving a set of non-commuting unitaries. We hope these results shed some light on M-QSP, giving directions for future works towards finding a full characterization of the multivariate polynomials.


\section*{Acknowledgements}
We would like to thank Yuki Ito for useful feedback on Theorem~\ref{thm:inapproximability-bounds}. The authors acknowledge support from the Swiss National Science Foundation (SNSF), project
No.\ 200020-214808.

\bibliographystyle{quantum}
\bibliography{refs}

\appendix
\section{Proof for more than two variables}
\label{apx:proof-m-variables}

For completeness, we provide a proof for Theorems~\ref{thm:condition-for-single-step} and~\ref{thm:sufficient-condition-full-decomposition} in the case of $m \ge 3$ variables. Analogously to what we did in two variables, we define a $(m+1)$-dimensional protocol, where the signal operator is
\begin{align*}
    \Tilde{W} = \diag(1, z_1, \ldots, z_m)
\end{align*}
and the processing operators are $A_k \in SU(m+1)$.
\begin{theorem}
    \label{thm:condition-for-single-step-mvars}
    Let $\ket{\gamma(\vv{z})} = \sum_{\vv{k}} \ket{\gamma_{\vv{k}}} \vv{z}^{\vv{k}}$ be a polynomial state of degree~$n$ in $m$ variables and $m+1$ dimensions. There exists $A_n \in SU(m+1)$ such that $\Tilde{W}^\dag A_n^\dag \ket{\gamma(\vv{z})}$ has degree $n - 1$ if and only if there exists an orthonormal basis $\{ \ket{\psi_j} \}_{0 \le j \le m}$ such that
    \begin{enumerate}[label=(\roman*)]
        \item $\braket{\psi_j}{\gamma_\vv{k}} = 0$, for $1 \le j \le m$, and every $\vv{k} \ge 0$ such that $k_j = 0$;
        \item $\braket{\psi_0}{\gamma_{\vv{k}}} = 0$ for every $\vv{k} \ge 0$ such that $\sum_j k_j = n$.
    \end{enumerate}
    If the $m+1$ endpoints $\{ \ket{\gamma_{\vv{e}_j}} \}_{0 \le j \le m}$ are non-zero, the above conditions are equivalent to
    \begin{enumerate}[label=(\roman*')]
        \item $\braket{\gamma_{n \vv{e}_j}}{\gamma_{\vv{k}}} = 0$, for $1 \le j \le m$, and every $\vv{k} \ge 0$ such that $k_j = 0$;
        \item $\braket{\gamma_{\vv{0}}}{\gamma_{\vv{k}}} = 0$, for every $\vv{k} \ge 0$ such that $\sum_j k_j = n$.
    \end{enumerate}
\end{theorem}
\noindent Condition (i) says that $\ket{\psi_j}$ --- which will be mapped to the subspace where $\Tilde{W}$ multiplies by $z_j$ --- must be orthogonal to the coefficients of the terms where $z_j$ already has degree $0$, as to avoid any negative coefficient. Similarly, the second condition ensures that the subspace that does not lower any degree is orthogonal to the vectors that has the maximum degree $n$, so that they are all lowered and the final polynomial can be of degree $n-1$.
\begin{proof}
    Following the same reasoning of Theorem~\ref{thm:condition-for-single-step}, $\ket{\gamma'(\vv{z})} = \Tilde{W}^\dag A_n^\dag \ket{\gamma(\vv{z})}$ is of degree $n-1$ if and only if $\ket{\gamma'_{\vv{k}}} = 0$ for every $\vv{k}$ of degree $n$, and every $\vv{k}$ that contains negative entries. Therefore, we need $A_n^\dag$ to satisfy
    \begin{itemize}
        \item $\bra{j} A_n^\dag \ket{\gamma_{\vv{k}}} = 0$ for every $1 \le j \le m$ and $\vv{k} \ge 0$ with $k_j = 0$, to avoid negative degrees;
        \item $\bra{0} A_n^\dag \ket{\gamma_{\vv{k}}} = 0$ for every $\vv{k} \ge 0$ with $\sum_j k_j = n$, to reduce the degree.
    \end{itemize}
    By conditions (i)-(ii), the choice $A_n : \ket{j} \mapsto \ket{\psi_j}$ suffices. The endpoints $\ket{\gamma_{\vv{0}}}, \ket{\gamma_{n\vv{e}_j}}$ need to satisfy $m$ out of the $m+1$ conditions, thus implying $\ket{\gamma_{n\vv{e}_j}} \propto \ket{\psi_j}$, while $\ket{\gamma_{\vv{0}}} \propto \ket{\psi_0}$, implying that (i)-(ii) and (i')-(ii') are equivalent whenever these $m+1$ endpoints are non-zero.
\end{proof}

\begin{theorem}
    \label{thm:sufficient-condition-full-decomposition-mvars}
    Let $\ket{\gamma(\vv{z})} = \sum_{k, h} \ket{\gamma_{\vv{k}}} \vv{z}^{\vv{k}}$ be a polynomial state of degree~$n$ in $m$ variables and $(m+1)$-dimensions such that the coefficient of $1, z_1^n, \ldots, z_m^n$ are all non-zero. Then there exist $A_0, \ldots, A_n \in SU(m+1)$ such that
    \begin{align*}
        A_n \Tilde{W} A_{n-1} \Tilde{W} \cdots \Tilde{W} A_0 \ket{0} = \ket{\gamma(\vv{z})}\ .
    \end{align*}
\end{theorem}
\begin{proof}
    For $n = 0$ the claim is trivial, since the polynomial state is a constant quantum state and can be chosen as $A_0 \ket{0}$.

    For $n > 0$, since the endpoints are non-zero, the conditions (i')-(ii') of Theorem~\ref{thm:condition-for-single-step-mvars} must hold in order for a degrading $A_n$ to exist. By writing down the normalization condition we obtain:
    \begin{align*}
        \braket{\gamma(\vv{z})}{\gamma(\vv{z})} = \sum_{\vv{k}, \vv{h}} \braket{\gamma_{\vv{k}}}{\gamma_{\vv{h}}} \vv{z}^{\vv{h} - \vv{k}} = \sum_{\vv{h}} \sum_{\vv{k}} \braket{\gamma_{\vv{k}}}{\gamma_{\vv{k} + \vv{h}}} \vv{z}^{\vv{h}} \stackrel{!}{\equiv} 1
    \end{align*}
    The coefficient of $\vv{z}^{\vv{h}}$ for $\vv{h} \neq 0$ must be zero. In particular
    \begin{itemize}
        \item by taking $\vv{h} = \vv{k}' - n \vv{e}_j$ for $\vv{k}' \ge 0$ with $k'_j = 0$, we get $\braket{\gamma_{n\vv{e}_j}}{\gamma_{\vv{k}'}} = 0$, condition (i');
        \item by taking $\vv{h}$ of degree $n$, we obtain $\braket{\gamma_{\vv{0}}}{\gamma_{\vv{h}}} = 0$, condition (ii').
    \end{itemize}
    This yields a $A_n$ that lower the degree of $\ket{\gamma'(\vv{z})} = \Tilde{W}^\dag A_n^\dag \ket{\gamma(\vv{z})}$, by Theorem~\ref{thm:condition-for-single-step-mvars}. Moreover, $\ket{\gamma'(\vv{z})}$ must have a non-zero coefficient in $z_j^{n-1}$ for each $j$, otherwise $A_n \Tilde{W} \ket{\gamma'(\vv{z})}$ would not be able to produce a term of degree $z_j^n$. The same reasoning applied to the constant term, and the claim follows by the induction step.
\end{proof}

\end{document}